\documentclass[runningheads]{llncs}

\usepackage[T1]{fontenc}
\usepackage{graphicx}

\usepackage{algorithm}
\usepackage{algpseudocode}

\usepackage{listings}   
\usepackage{xcolor}     
\usepackage{hyperref}   

\usepackage{cite}

\usepackage{comment}
\usepackage{amsmath,amssymb}
\usepackage{stmaryrd}
\usepackage[colorinlistoftodos, prependcaption, textsize=footnotesize]{todonotes}

\usepackage{mathtools}
\usepackage{listings}
\usepackage{xcolor}
\usepackage[scaled=0.85]{beramono} 
\usepackage{stmaryrd}

\lstdefinelanguage{Maude}{
    morekeywords={rl,eq,ceq,crl,mod,fmod,endm,endfm,sort,sorts,subsort,subsorts,op,ops,var,vars,mb,cmb,nonexec,narrowing,metadata},
    sensitive=true,
    morecomment=[l]{---},
    morecomment=[s]{***(}{)},
    morestring=[b]"
}

\usepackage{amsmath,amssymb}
\usepackage{mathpartir} 
\usepackage{cleveref} 

\newcommand{\ms}{,\;}

\newcommand{\state}[2]{\langle\,#1\,\Vert\,#2\,\rangle}
\newcommand{\cstate}[3]{\langle\,#1\,\Vert\,#2\,\Vert\,#3\,\rangle}

\newcommand{\graph}[1]{\mathsf{graph}(#1)}
\newcommand{\status}[1]{\mathsf{phase}(#1)}
\newcommand{\freshIds}[1]{\mathsf{freshIds}(#1)}
\newcommand{\allIds}[1]{\mathsf{allIds}(#1)}
\newcommand{\missing}[2]{\mathsf{missing}(#1, #2)}
\newcommand{\refPool}[1]{\mathsf{refPool}(#1)}
\newcommand{\need}[2]{\mathsf{need}(#1, #2)}

\newcommand{\objectBuild}{\mathsf{objectBuild}}
\newcommand{\refBuild}{\mathsf{refBuild}}

\newcommand{\nil}{\mathsf{nil}}

\newcommand{\refIdx}[5]{\mathsf{refIdx}(#1,#2,#3,#4,#5)}

\newcommand{\NF}{\mathsf{NF}}

\newtheorem{assumption}{Assumption}

\newcommand{\simpB}{\mathsf{simpB}}

\newcommand{\specInvOnCreate}{\mathsf{specInvOnCreate}}
\newcommand{\specInvOnSetRef}{\mathsf{specInvOnSetRef}}

\newcommand{\initAttrs}{\mathsf{initAttrs}}

\newcommand{\setRef}{\mathsf{setRef}}

\newcommand{\choiceRef}{\mathsf{choiceRef}}

\newcommand{\Accept}{\mathit{Accept}}
\newcommand{\SAT}{\mathit{sat}}

\newcommand{\shapeInd}{\eta}

\newcommand{\confSort}{\text{\texttt{Configuration}}}

\newcommand{\phiC}{\phi}      
\newcommand{\K}{K}            
\newcommand{\R}{R}            
\newcommand{\conf}{C}            
\newcommand{\obj}{O}            
\newcommand{\os}{\mathit{OS}}          
\newcommand{\ol}{\mathit{OL}}          
\newcommand{\rt}{\mathit{RT}}          
\newcommand{\rtList}{\mathit{RTL}}     
\newcommand{\lb}{\ell}        
\newcommand{\ub}{u}           

\newcommand{\monus}{\mathbin{\ominus}}

\newcommand{\Inv}{\mathsf{Inv}}

\newcommand{\stepRel}{\;\Rightarrow\;}

\ifx\definition\undefined
\newtheorem{definition}{Definition}[section]
\fi
\ifx\theorem\undefined
\newtheorem{theorem}{Theorem}[section]
\fi
\ifx\lemma\undefined
\newtheorem{lemma}[theorem]{Lemma}
\fi
\ifx\remark\undefined
\newtheorem{remark}{Remark}[section]
\fi

\begin{document}

\title{Bounded Structural Model Finding with Symbolic Data Constraints}
\author{Artur Boronat\inst{1}\thanks{This work was supported by the European Space Agency under ESA Contract 4000133105/20/NL/AF (P-STEP project).}}
\authorrunning{Artur Boronat}

\institute{University of Leicester, Leicester, UK\\
\email{artur.boronat@leicester.ac.uk}
}

\maketitle

\begin{abstract}
Bounded model finding is a key technique for validating software designs, usually obtained by translating high-level specifications into SAT/SMT problems. Although effective, such translations introduce a semantic gap and a dependency on external tools. We present the \emph{Maude Model Finder} (MMF), a native approach that brings bounded model finding to the Maude rewriting logic framework. MMF provides a schema-parametric engine that repurposes symbolic reachability for structural solving, generating finite object configurations from class declarations and graph and data constraints without bespoke generators. Technically, MMF explores a constrained symbolic rewriting system over runtime states modulo an equational theory; SMT is used only to prune states by constraint satisfiability and to discharge entailment checks for subsumption and folding. We contribute a bounded, obligation-driven calculus that separates object creation from reference assignment and supports symmetry reduction by folding via Maude's ACU matching. We establish termination, soundness, and completeness of the bounded construction within the declared bounds, and justify folding via a coverage-preserving subsumption test. We focus on the calculus and its properties, illustrating it on running examples supported by a Maude prototype.

\end{abstract}

\section{Introduction}
\label{sec:introduction}

Maude is a high-performance specification and analysis environment based on rewriting logic and membership equational logic \cite{clavel_2007_all_about_maude,durn_2020_programming_symbolic_computation}. Its mainstream analysis capabilities target \emph{dynamic} properties, including reachability search and LTL model checking for finite-state systems \cite{bae_2014_rewritingbased_model_checking}, and symbolic analysis for richer state spaces, including symbolic reachability and rewriting modulo SMT \cite{meseguer_2007_symbolic_reachability_analysis,bae_2019_symbolic_state_space}. These strengths make Maude an attractive semantic platform, but they do not directly support a complementary task that is central in modelling language engineering: bounded (structural) model finding. 

Bounded model finding aims to generate finite object graphs within declared bounds that satisfy (i) class and role\footnote{A \emph{role} is a typed reference between classes with a declared multiplicity range.} multiplicities, and (ii) structural and data constraints. Such bounded instances are essential to debug metamodels and Object Constraint Language (OCL) constraints, validate design assumptions, and obtain concrete witnesses and counterexamples. Tools such as Alloy and the USE validator provide this capability through SAT/SMT-based instance generation \cite{jackson_2019_alloy_language_tool,kuhlmann_2011_extensive_validation_ocl}. In Maude, while one can execute specifications and explore behaviours, there is no out-of-the-box, schema-parametric mechanism dedicated to bounded \emph{structural} instance generation for object configurations. 
%
In our earlier work on model subtyping with OCL constraints and safe reuse, bounded witnesses were obtained via the USE validator, using its Kodkod-based back end to check satisfiability of OCL constraints \cite{BoronatSLE17,boronat_safe_2023}. This paper internalises bounded witness and counterexample generation within Maude.

 We present the \emph{Maude Model Finder} (MMF), a schema-parametric bounded model finding approach for Maude object configurations. MMF formulates bounded model construction as exploration of a constrained symbolic rewriting system over \emph{runtime states} modulo oriented equations $E$ and algebraic axioms $B$ (such as associativity, commutativity, and identity). SMT is used only as a background solver to prune states by constraint satisfiability and to discharge entailment checks needed for subsumption and folding \cite{bae_2019_symbolic_state_space}. Concretely, we contribute: (1) a bounded, obligation-driven construction calculus that separates object allocation from reference assignment and incrementally accumulates symbolic data constraints; (2) redundancy reduction via folding, combining matching modulo ACU, shape indexing, and SMT-backed entailment; and (3) meta-theoretic guarantees for the bounded calculus, including termination, soundness, and completeness within the declared bounds, and soundness of folding with respect to semantic coverage. We focus on the calculus and its properties, illustrating it on running examples supported by a Maude prototype.

The paper is organised as follows. Section~\ref{sec:mmf-running-example} presents the MMF specification style through a running example. Sections~\ref{sec:mmf-foundations} and \ref{sec:mmf-semantics} introduce the semantic foundations and the symbolic construction calculus. Section~\ref{sec:mmf-reachability} describes the reachability engine and its reduction pipeline, and Section~\ref{sec:mmf-proofs} summarises the correctness results with full proofs in Appendix~\ref{sec:app-proofs}. Section~\ref{sec:evaluation} reports a proof-of-concept evaluation, followed by related work and conclusions.

\section{Specification with the Maude Model Finder (MMF)}
\label{sec:mmf-running-example}

We illustrate MMF on a small organizational schema with classes, references, bounds, and symbolic data constraints. An MMF specification describes a bounded family of admissible object graphs in three parts: (i) a \emph{schema} declaring classes and roles, (ii) a \emph{scope} fixing finite class and role bounds, and (iii) \emph{constraints}, split into structural predicates over graphs and data invariants.

MMF generates models by \emph{constrained symbolic rewriting} over \emph{runtime states} $\state{\conf}{\phiC}$, where $\conf$ is an object-graph configuration (matched modulo ACU axioms) and $\phiC$ is an accumulated Boolean constraint over symbolic attribute variables. Exploration proceeds in two phases: an object allocation phase creates objects to meet class bounds, followed by a reference assignment phase that enumerates admissible role values within role bounds. Branches are pruned as soon as either (i) a monotonic structural error pattern is detected in the partial graph, or (ii) the constraint becomes SMT-unsatisfiable; non-monotonic global constraints are checked only on completed candidates (normal forms).

\paragraph{Schema and bounds.}
The running example has companies, employees, and projects: a \texttt{Company} has a mandatory CEO and a set of projects; an \texttt{Employee} may have a manager; a \texttt{Project} has a non-empty set of members. The schema is written in Maude's object-oriented style (Listing~\ref{lst:org-proj-schema}); role multiplicities are reflected by carrier sorts (e.g., \texttt{Oid}, \texttt{Maybe\{Oid\}}, \texttt{Set\{Oid\}}, \texttt{NeSet\{Oid\}}), while the effective branching space is fixed by explicit bounds.
Listing~\ref{lst:org-proj-bounds} declares a finite scope via a constant \texttt{boundsDecl}: \texttt{classB} fixes class cardinalities and \texttt{roleB} constrains, per role, the admissible target class and multiplicity range.

\begin{lstlisting}[language=Maude,caption={Schema for the CEO example.},label={lst:org-proj-schema}]
omod ORG-PROJ-SCHEMA is
  protecting MMF .
  class Company  | ceo : Oid , projects : Set{Oid} .
  class Project  | members : NeSet{Oid} .
  class Employee | manager : Maybe{Oid} .
endom
\end{lstlisting}

\begin{lstlisting}[language=Maude,caption={Example bounds for the CEO running example.},label={lst:org-proj-bounds}]
eq boundsDecl =
  classB(Company,  1, 1) ;
  classB(Employee, 2, 2) ;
  classB(Project,  0, 2) ;
  roleB(Company,  ceo,      Employee, 1, 1) ;
  roleB(Company,  projects, Project,  0, 2) ;
  roleB(Employee, manager,  Employee, 0, 1) ;
  roleB(Project,  members,  Employee, 1, 2) .
\end{lstlisting}

\paragraph{Structural constraints.}
MMF supports two structural error predicates.
First, \texttt{specErrorPredGraphPartial} is evaluated on \emph{partial} graphs to enable early pruning. It includes intrinsic well-formedness checks (e.g., unique object identifiers, no dangling edges, and adherence to upper bounds) and can be extended with user-defined monotonic error patterns; once triggered, such errors cannot be repaired by further refinement.
Second, \texttt{specErrorPredGraph} is evaluated only on completed candidates (normal forms), to express non-monotonic constraints that need the full graph. In this example, we forbid a manager for the CEO and require the management relation to be acyclic (Listing~\ref{lst:org-proj-graph-constraints}).
Both predicates return a value of sort \texttt{TruthPred}, whose sole constructor is \texttt{tt}. \texttt{TruthPred} is a dedicated sort distinct from Maude's \texttt{Bool}: error predicates are structurally matched against \texttt{tt} to detect violations, whereas \texttt{Bool} is reserved for SMT constraint terms that are passed to the solver. Keeping the two sorts separate avoids ambiguity when Boolean SMT formulas are manipulated as first-class terms alongside structural graph checks.

\begin{lstlisting}[language=Maude,caption={Global constraints (cycle detection helper omitted).},label={lst:org-proj-graph-constraints}]
op ceoHasManager : Configuration ~> TruthPred .
var C E1 : Oid .  var E2 : Maybe{Oid} .
var CAS AS : AttributeSet .  var CF : Configuration .
eq ceoHasManager(< C : Company | ceo : E1, CAS >
                 < E1 : Employee | manager : E2, AS > CF) = tt .
eq specErrorPredGraph(CF) =
  ceoHasManager(CF) ttor managerCyclic(CF) .
  --- ttor: short-circuit disjunction (tt-or) over TruthPred; managerCyclic: DFS cycle check (omitted)
\end{lstlisting}

\paragraph{Data invariants via constraints.}
Data-level properties are expressed by equational hooks\footnote{We use \emph{equational hook} to mean a Maude equation, named by convention, that fires on a specific construction event and returns a Boolean constraint contribution. This is distinct from Maude's built-in hook mechanism for foreign function interfaces.} that contribute to the constraint $\phiC$. The hook \texttt{specInvOnCreate} is applied when an object is created; \texttt{specInvOnSetRef} is applied when a role value is committed. In the example, each employee is given a symbolic integer \texttt{level} constrained to a finite range; managers must have strictly higher rank (smaller \texttt{level}), and the CEO must be at level 0 (Listing~\ref{lst:org-proj-hooks}). Here \texttt{i(O,\,AN)}, \texttt{b(O,\,AN)}, and \texttt{r(O,\,AN)} denote the symbolic SMT variables for the integer-, Boolean-, and real-valued attribute \texttt{AN} of object \texttt{O}, respectively (formalised in Section~\ref{sec:mmf-semantics}).

\begin{lstlisting}[language=Maude,caption={Equational constraints for the CEO example.},label={lst:org-proj-hooks}]
var E C : Oid .  var M : Maybe{Oid} .
eq specInvOnCreate(Employee, E) =
  i(E, level) >= (0).Integer and i(E, level) < (3).Integer .
ceq specInvOnSetRef(Employee, manager, E, M) =
  i(M, level) < i(E, level)
if M :: Oid .
eq specInvOnSetRef(Employee, manager, E, empty) = (true).Boolean .
eq specInvOnSetRef(Company, ceo, C, E) =
  i(E, level) === (0).Integer .
\end{lstlisting}

\paragraph{Execution modes.}
An MMF specification consists of the schema, \texttt{boundsDecl}, and the structural and data constraints above. The \texttt{find} mode enumerates valid models with \texttt{findAll}, while \texttt{findFirst} searches for a single witness:
\begin{lstlisting}[language=Maude]
reduce in MF-DRIVER-META : findAll(initModel(find)) .
\end{lstlisting}
\begin{lstlisting}[language=Maude]
reduce in MF-DRIVER-META : findFirst(initModel(find)) .
\end{lstlisting}
If \texttt{findFirst} yields no solution (or \texttt{findAll} an empty set), the specification is unsatisfiable within the declared bounds.

\paragraph{Counterexample search (\texttt{check}).}
The \texttt{check} mode reuses the same bounded construction semantics and pruning, but it searches for a counterexample instead of valid models.
A counterexample is an admissible normal form that violates a user-declared mixed property combining a graph pattern and attribute constraints.
Graph patterns are executable predicates over the graph component and, when monotonic, can be evaluated on partial graphs to prune early.
Attribute violations are searched by conjoining the negation of a constraint into \(\phiC\); a satisfiable final constraint witnesses a violating valuation.
We illustrate \texttt{check} with two small properties for the CEO example.
The first property is structural and detects a non-CEO employee without a manager (Listing~\ref{lst:org-prop-graph}).

\begin{lstlisting}[language=Maude,caption={Graph property for \texttt{check}: non-CEO employee without manager.},label={lst:org-prop-graph}]
op nonCeoWithoutManager : Configuration ~> TruthPred .
ceq nonCeoWithoutManager(< C : Company  | ceo : CEO, CAS >
     < E : Employee | manager : empty, AS > CF) = tt if E =/= CEO .
eq errorPredPropGraph(CF) = nonCeoWithoutManager(CF) .
\end{lstlisting}

Semantically, this graph property corresponds to an invariant $\Inv^{\mathsf{prop}}_{\mathsf{graph}}$.
Operationally, \texttt{check} uses \texttt{errorPredPropGraph} as an executable error predicate, which returns \texttt{tt} precisely when $\Inv^{\mathsf{prop}}_{\mathsf{graph}}$ is violated.

The second is an attribute property that negates the intended level range and searches for a satisfying violation (Listing~\ref{lst:org-prop-smt}).

\begin{lstlisting}[language=Maude,caption={Attribute-level property for \texttt{check}: level range (searched as a violation).},label={lst:org-prop-smt}]
op propInvOnCreate : ClassId Oid -> Boolean .
eq propInvOnCreate(Employee, E) =
  i(E, level) >= (0).Integer and i(E, level) < (3).Integer .
\end{lstlisting}


\section{Semantic Foundations}
\label{sec:mmf-foundations}

We separate the \emph{declarative meaning} of a specification from its
\emph{operationalization} by the MMF calculus.

\begin{definition}[Core model specification]
\label{def:core-model-spec}
 A \emph{core model specification} $\mathcal{M} = (\mathcal{S}, \Inv^{\mathsf{part}}_{\mathsf{graph}}, \Inv^{\mathsf{full}}_{\mathsf{graph}}, \Inv_{\mathsf{attr}})$ defines the validity criteria for constrained model instances.
 Here $\mathcal{S}$ is a Maude schema module declaring classes, attributes, and subclass relations; $\Inv^{\mathsf{part}}_{\mathsf{graph}}$ and $\Inv^{\mathsf{full}}_{\mathsf{graph}}$ are graph invariants for partial and completed graphs, respectively, with $\Inv^{\mathsf{part}}_{\mathsf{graph}}$ required to be monotonic; and $\Inv_{\mathsf{attr}}$ is the attribute invariant.
\end{definition}

We view
$\Inv^{\mathsf{part}}_{\mathsf{graph}}$ as comprising an \emph{intrinsic well-formedness}
condition together with schema-provided monotonic constraints. Concretely, the intrinsic
component requires that (i) object identifiers are unique, (ii) for every reference role,
the current reference set does not exceed the declared upper multiplicity, and (iii) the
given bounds are consistent, in the sense that every target class has an upper scope large
enough to satisfy the maximum lower multiplicity demanded by any incoming role.

 Let $G$ range over Maude object configurations---multisets of object records $\langle o : C \mid \mathrm{AS} \rangle$ under multiset union modulo ACU. We write $G : \mathcal{S}$ to denote that $G$ \emph{conforms} to the schema, meaning that every object and reference in $G$ respects the declarations and typing in $\mathcal{S}$. 

\begin{remark}[Invariants vs.\ error predicates]
\label{rem:inv-vs-error}
At the semantic level, $\Inv^{\mathsf{part}}_{\mathsf{graph}}$ and
$\Inv^{\mathsf{full}}_{\mathsf{graph}}$ are invariants.
Operationally, MMF's \emph{error predicates} return $\mathsf{tt}$
when an invariant is \emph{violated}. For partial graphs, this check
combines built-in and schema-provided monotonic constraints. Thus, pruning
checks the \emph{negation} of semantic invariants.
\end{remark}

\begin{definition}[Bounded model specification]
\label{def:bounded-model-spec}
A \emph{bounded model specification} $\mathcal{M}_b = (\mathcal{M}, b)$ consists of a
core specification $\mathcal{M}$ together with a scope $b$.
The scope assigns to each class $C$ a finite cardinality interval $[\ell_C, u_C]$
and to each role $r : C \to D$ a finite multiplicity interval $[\ell_r, u_r]$,
bounding the number of targets that each object of class $C$ may have along $r$.
\end{definition}

We write $G : \mathcal{S}_b$ iff $G : \mathcal{S}$ and all bounds in $b$ hold in $G$:
each class $C$ occurs in $G$ with cardinality in $[\ell_C, u_C]$, and for each
object $o : C$ the set of target object identifiers assigned to role $r$ of $o$,
denoted $o.r$, has cardinality in $[\ell_r, u_r]$.

\begin{definition}[Constrained model instance]
\label{def:constrained-model}
 A \emph{constrained model instance} is a pair $M=\langle G,\rho\rangle$ where $G$ is a \emph{graph topology} (a finite set of typed object records with role assignments, i.e., an object configuration $G : \mathcal{S}$) and $\rho$ is a ground valuation for the attribute-level symbols occurring in $G$ and in the attribute invariant $\Inv_{\mathsf{attr}}(G)$.
\end{definition}

\begin{definition}[Conformance]
\label{def:conformance}
 Let $\mathcal{M}_b$ be a bounded specification. A constrained model instance $M=\langle G,\rho\rangle$ \emph{conforms} to $\mathcal{M}_b$, written $M : \mathcal{M}_b$, iff:
 \[
 \begin{aligned}
 & G : \mathcal{S}_b
 \ \wedge\ (G \models \Inv^{\mathsf{part}}_{\mathsf{graph}})
 \ \wedge\ (G \models \Inv^{\mathsf{full}}_{\mathsf{graph}})
 \ \wedge\ (\rho \models \Inv_{\mathsf{attr}}(G)).
 \end{aligned}
 \]
\end{definition}

\begin{definition}[Bounded semantics]
\label{def:bounded-semantics}
 Define
 \(
 \llbracket \mathcal{M}_b \rrbracket
 \;=\;
 \{\, M \mid M : \mathcal{M}_b \,\}.
 \)
\end{definition}


\section{Symbolic State and Transition Semantics}
\label{sec:mmf-semantics}

 We now turn to the formal machinery enabling automatic model finding, the \emph{MMF calculus}. Unlike operational
 semantics that evolves a system over time, MMF employs a \emph{constructive} symbolic
 transition system where configurations denote partial graphs and transitions represent refinement
 decisions. This design incrementally builds the graph, explicitly separating
 \emph{existence} from \emph{topology}.

This section formalizes the symbolic transition semantics of MMF. We first define
runtime states in Section~\ref{sec:mmf-calculus-states}, which encapsulate the
partial graph and symbolic constraints. We then detail the two-phase construction process:
the object generation phase in Section~\ref{sec:mmf-object-generation}, which allocates
objects and their data attributes, and the graph creation phase in
Section~\ref{sec:mmf-reference-generation}, which enumerates reference topologies over
the fixed object universe.


\subsection{States in the MMF Runtime}
\label{sec:mmf-calculus-states}


The MMF calculus operationalises the bounded satisfiability problem for a bounded model specification $\mathcal{M}_b$ as a finite reachability search over runtime states.
Let \confSort{} be the Maude configuration sort, i.e., multisets of objects and control tokens
under multiset union modulo ACU. We write $\conf, \conf' : \confSort$.
We write $\graph{G}$ for the object-graph fragment (a multiset of object records),
and $\phiC$ for Boolean constraints over the background SMT theories.
 Class and role identifiers are $\K$ and $\R$; object identifiers $\obj$
 are drawn from a finite pool bounded by $b$.

\begin{definition}[Runtime states and graph projection]
\label{def:runtime-states}
A \emph{runtime state} is a constrained configuration $S \equiv \state{\conf}{\phiC}$,
where $\conf$ is a multiset of objects including $\graph{G}$ and control tokens modulo ACU, and $\phiC$
is a Boolean constraint. Let $\pi(\conf)$ be the \emph{graph projection} that erases
all runtime tokens from $\conf$ and keeps only the object records.
\end{definition}

 Semantically, $S$ denotes a set of constrained model instances $M=\langle \pi(\conf),\rho\rangle$ such that $\rho \models \phiC$; its denotation is given in Definition~\ref{def:state-semantics}. Object identifiers are treated as fixed uninterpreted symbols during construction; only attribute-level symbols are assigned values by the constraint valuations.
A configuration $\conf$ packages both the evolving graph and bounded construction control.
Its core tokens include a \emph{phase marker} $\status{q}$ with $q \in \{\objectBuild, \refBuild\}$;
an \emph{identifier universe} $\allIds{\os}$ indexed by class;
an \emph{allocation cursor} $\freshIds{\ol}$ and the set of allocated
identifiers $\mathsf{takenIds}(\mathit{TI})$;
\emph{lower-bound obligations} $\missing{\K}{n}$ and \emph{feasibility
obligations} $\mathsf{need}(\K, n)$; and
a \emph{reference-task pool} $\refPool{\rtList}$ discharged in phase $\refBuild$.
These tokens are operational devices: they ensure boundedness, phase discipline, and
exhaustive enumeration.

\begin{remark}[Role of runtime tokens]
\label{rem:runtime-tokens}
Runtime tokens (identifier pools, obligation counters, task lists, phase markers)
are part of $\conf$ because they drive the operationalization, support termination,
and enable completeness by systematic enumeration. They do \emph{not} appear in
the semantic satisfaction relation $\models$; the link is mediated by the projection
$\pi$ and the state denotation $\llbracket \cdot \rrbracket$.
\end{remark}


\subsection{Overview of the MMF Construction Calculus}
\label{sec:mmf-calculus-overview}

The MMF calculus is \emph{two-phase}: in phase $\objectBuild$ it allocates a bounded set
of objects, and in phase $\refBuild$ it assigns reference-valued attributes. Let $B$ be the structural axioms involved in the constructors of $\mathcal{M}$ and runtime tokens. Each one-step
move is a rule application on the runtime configuration $\conf$ modulo $B$, followed by
equational normalization by the oriented equations $\vec{E}$ modulo $B$.

 We distinguish three layers of rewriting: a \emph{raw} rule step operates on the configuration $\conf$ (modifying the graph topology $G$); a \emph{normalized} step applies equational simplification; and a \emph{constrained} step lifts the normalized structural transition to the level of runtime states by accumulating the attribute constraints implied by the structural change.

\begin{definition}[Raw rewriting step]
\label{def:raw-step}
We write $\conf \to^{\ell,\sigma}_{R/B} \conf_1$ iff there exists a rule in $R$ with
label $\ell$ and a substitution $\sigma$ (obtained by matching modulo $B$) such that
applying the instantiated rule $\ell\sigma$ rewrites $\conf$ to $\conf_1$ modulo $B$.
\end{definition}

\begin{definition}[Normalized rewriting step]
\label{def:norm-step}
We write $\conf \Rightarrow^{\ell,\sigma}_{R,\vec{E}/B} \conf'$ iff there exists an
intermediate $\conf_1$ such that $\conf \to^{\ell,\sigma}_{R/B} \conf_1$ and
$\conf' = \conf_1 \downarrow_{\vec{E}/B}$, where $\conf_{1} \downarrow_{\vec{E}/B}$ denotes the canonical form of $\conf_{1}$ with respect to the oriented equations $\vec{E}$ modulo axioms $B$.
\end{definition}

\begin{remark}[Maude implementation of rewriting steps]
\label{rem:maude-steps}
The raw step $\conf \to^{\ell,\sigma}_{R/B} \conf_1$ is realised via Maude's reflective \texttt{metaApply}, which applies a single rewrite rule and returns the result modulo $B$.
Normalization $\downarrow_{\vec{E}/B}$ corresponds to \texttt{metaReduce}, which computes the $\vec{E}/B$-canonical form.
These are the standard reflective operations used throughout the MMF engine (Section~\ref{sec:mmf-reachability}).
\end{remark}

Only two rule families add user attribute constraints to the guard: object creation
(\textsc{Obj-Gen}) and reference commitment (\textsc{Ref-Choose}). We factor these
event-specific contributions through two hook shorthands:
\begin{align*}
\mathsf{hook}_{\mathsf{create}}(K,o) &\triangleq \specInvOnCreate(K,o)\\
\mathsf{hook}_{\mathsf{setref}}(K,R,o,v) &\triangleq \specInvOnSetRef(K,R,o,v)
\end{align*}
and define the auxiliary function $\mathsf{hook}_\ell(\sigma)$ to capture the rule-indexed
injection logic:
\[
\mathsf{hook}_\ell(\sigma) \triangleq \begin{cases}
\mathsf{hook}_{\mathsf{create}}(K\sigma, \obj\sigma) & \text{if } \ell = \textsc{Obj-Gen} \\
\mathsf{hook}_{\mathsf{setref}}(K\sigma, \R\sigma, \obj\sigma, v\sigma) & \text{if } \ell = \textsc{Ref-Choose} \\
\mathsf{true} & \text{otherwise}
\end{cases}
\]

\begin{definition}[Constrained rewriting step]
\label{def:constrained-step}
We write $$\state{\conf}{\phiC} \Rightarrow^{\ell,\sigma}_{R,\vec{E}/B} \state{\conf'}{\phiC'}\qquad \text{if} \qquad
\conf \Rightarrow^{\ell,\sigma}_{R,\vec{E}/B} \conf',$$ and
the constraint is strengthened by the corresponding hook and normalized:
$\phiC' = \simpB\bigl(\phiC \wedge \mathsf{hook}_\ell(\sigma)\bigr) \downarrow_{\vec{E}/B}$\footnote{The $\simpB$ function implements a deterministic simplification strategy that normalizes Boolean constraints by flattening associative-commutative operators, ordering literals, pushing negations into relational atoms, performing arithmetic constant folding, and eliminating identity elements.}.
\end{definition}

\begin{remark}[Contextual conventions]
\label{rem:context-conventions}
Unless stated otherwise, all runtime configurations and constraints are maintained in
$\vec{E}/B$-canonical form.
In the rule presentations below, we write $\state{\conf}{\phiC} \Rightarrow \state{\conf'}{\phiC'}$
as shorthand for the constrained rewriting step $\state{\conf}{\phiC} \Rightarrow^{\ell,\sigma}_{R,\vec{E}/B}
\state{\conf'}{\phiC'}$; the rule label $\ell$, matching substitution $\sigma$, and theory parameters
$R, \vec{E}, B$ are implicit from context.
We write $S \stepRel S'$ when $S \Rightarrow^{\ell,\sigma}_{R,\vec{E}/B} S'$ holds for some $\ell$ and $\sigma$
(Definition~\ref{def:constrained-step}); in \texttt{check} mode, the step additionally updates the witness component $\psi_P$ as described below.
\end{remark}

Let $\mathsf{init}(\mathcal{M}_b, \mathsf{mode})$ be the canonical initial runtime state.
The configuration $\conf_0$ contains $\status{\objectBuild}$, an identifier universe
$\allIds{\os}$ and allocation list $\freshIds{\ol}$, where $\ol$ is a fixed
linearization of the identifiers in $\os$ induced by the class order
$\mathsf{classList}$ and the per-class upper bounds in $b$.
Its head determines the next identifier to be consumed and hence the canonical allocation order.
It also contains the initial bounded-work obligations: for every class $C$ in the schema, $\conf_0$ includes a token $\missing{C}{\ell_C}$, where $\ell_C$ is the class lower bound from $b$, and a token $\need{C}{\mathrm{inNeed}(C)}$, where $\mathrm{inNeed}(C)$ is the maximum lower multiplicity among roles targeting $C$.
Finally, $\phi_0 = \mathsf{true}$.
If $\mathsf{mode} = \texttt{find}$, then $\mathsf{init}(\mathcal{M}_b, \texttt{find}) = \state{\conf_0}{\phi_0}$.
If $\mathsf{mode} = \texttt{check}$, then $\mathsf{init}(\mathcal{M}_b, \texttt{check}) = \cstate{\conf_0}{\phi_0}{\psi_0}$ with $\psi_0 = \mathsf{false}$.
In \texttt{check} mode, $\phiC$ accumulates hard constraints, as usual, but $\psi_P$ accumulates an existential violation witness.

\paragraph{Property instrumentation (\texttt{check}).}
In \texttt{check} mode, the calculus is additionally parameterised by a mixed property with a graph component and an attribute component.
Semantically, the graph component is an invariant $\Inv^{\mathsf{prop}}_{\mathsf{graph}}$ evaluated on projected graphs.
Operationally, the implementation provides an error predicate $\mathsf{errorPredPropGraph}$ such that $\mathsf{errorPredPropGraph}(G) = \mathsf{tt}$ iff $G \not\models \Inv^{\mathsf{prop}}_{\mathsf{graph}}$.
The attribute component is specified by hook equations $\mathsf{propInvOnCreate}$ and $\mathsf{propInvOnSetRef}$ (Section~\ref{sec:mmf-running-example}).
Operationally, \texttt{check} reuses the same construction rules and the same hard-constraint accumulation.
In addition, it records whether some property violation can be witnessed by maintaining a separate Boolean witness component $\psi_P$ that accumulates violation evidence disjunctively across partial states.
Formally, for a rule instance with label $\ell$ and matching substitution $\sigma$, define
\[
\mathsf{propViol}_\ell(\sigma) \triangleq
\begin{cases}
\neg\,\mathsf{propInvOnCreate}(K\sigma, \obj\sigma) & \text{if } \ell = \textsc{Obj-Gen} \\
\neg\,\mathsf{propInvOnSetRef}(K\sigma, \R\sigma, \obj\sigma, v\sigma) & \text{if } \ell = \textsc{Ref-Choose} \\
\mathsf{false} & \text{otherwise.}
\end{cases}
\]
In \texttt{check} mode, the hard constraint is updated as usual.
The witness is updated by disjunction:
$\psi_P' = \simpB\bigl(\psi_P \vee \mathsf{propViol}_\ell(\sigma)\bigr)$.
Accordingly, in \texttt{check} mode we write
\[
\cstate{\conf}{\phiC}{\psi_P} \Rightarrow^{\ell,\sigma}_{R,\vec{E}/B} \cstate{\conf'}{\phiC'}{\psi_P'}
\]
iff $\conf \Rightarrow^{\ell,\sigma}_{R,\vec{E}/B} \conf'$ and the above two updates hold.

The calculus runs in two modes: \texttt{find} enumerates valid models satisfying all structural and data invariants, while \texttt{check} searches for a counterexample to a user-declared property.
Normal-form acceptance is defined accordingly.

\begin{definition}[Normal forms, admissibility, acceptance]
\label{def:nf-admissible-accept}
Normal form $\NF(S)$ and $\mathsf{Admissible}(S)$ are defined as follows:
\begin{itemize}
    \item $\NF(S)$ holds iff no calculus rule applies to $S$, i.e., iff $\neg\exists S'.\; S \stepRel S'$.
    \item $\mathsf{Admissible}(S)$ holds iff $\pi(\conf) \models \Inv^{\mathsf{part}}_{\mathsf{graph}}$ and $\SAT(\phiC)$.
    \item $\mathsf{Accept}_{\mathtt{find}}(S)$ holds iff $\NF(S) \wedge \mathsf{Admissible}(S) \wedge \bigl(\pi(\conf) \models \Inv^{\mathsf{full}}_{\mathsf{graph}}\bigr)$.
    \item $\mathsf{Accept}_{\mathtt{check}}(S)$ holds iff $\mathsf{Accept}_{\mathtt{find}}(S)$ and, additionally,
    \[
    \pi(\conf) \not\models \Inv^{\mathsf{prop}}_{\mathsf{graph}} \;\vee\; \SAT(\phiC \wedge \psi_P).
    \]
\end{itemize}
\end{definition}

\begin{remark}[Counterexamples in \texttt{check} mode]
\label{rem:check-counterexamples}
In \texttt{check} mode, hard constraints are accumulated in $\phiC$ and define the denotation $\llbracket S \rrbracket$.
Property violations are accumulated in the witness component $\psi_P$ by disjunction.
Therefore, if $\SAT(\phiC \wedge \psi_P)$ holds, there exists a valuation $\rho$ such that $\rho \models \phiC$ and $\rho \models \psi_P$.
Then $\langle \pi(\conf), \rho \rangle \in \llbracket S \rrbracket$.
Moreover, under $\rho$ the attribute-level property is violated at least once.
Graph-level violations are witnessed independently by $\pi(\conf) \not\models \Inv^{\mathsf{prop}}_{\mathsf{graph}}$.
\end{remark}

\begin{definition}[State semantics]
\label{def:state-semantics}
Let $\rho$ range over valuations that assign \emph{ground} values to all attribute-level symbols occurring in $\phiC$ (e.g., symbols of the form $i(O,AN)$, $r(O,AN)$, $b(O,AN)$ for an object identifier $O$ and attribute name $AN$), while leaving the object identifiers occurring in $\conf$ unchanged.
Define the \emph{semantics} of a runtime state by
\[
\llbracket \state{\conf}{\phiC} \rrbracket
\;=\;
\{\, \langle \pi(\conf), \rho \rangle \mid \rho \models \phiC\,\}
\;=\;
\llbracket \cstate{\conf}{\phiC}{\psi_P} \rrbracket.
\]
\end{definition}

 Intuitively, $\llbracket \state{\conf}{\phiC} \rrbracket$ is the set of constrained model instances $M=\langle \pi(\conf), \rho \rangle$ such that $\rho$ satisfies the accumulated constraint $\phiC$.
 The projection $\pi(\conf)$ provides the graph topology and the object identifiers, which are treated as fixed uninterpreted symbols during construction.
 The valuation $\rho$ supplies ground values for the attribute-level symbols occurring in $\phiC$.
In \texttt{check} mode, the witness component $\psi_P$ affects acceptance but not denotation.


\subsection{Object Generation Phase ($\objectBuild$)}
\label{sec:mmf-object-generation}

Phase $\objectBuild$ allocates fresh objects and initializes reference-assignment obligations.
We assume mode \texttt{find} and a configuration containing $\status{\objectBuild}$.
The transition relation $\stepRel$ is given by constrained rewriting rules over runtime states.
Rule \textsc{Obj-Gen} matches the configuration modulo ACU to extract the next identifier
$\obj_\K$ from $\freshIds{\obj_\K :: \ol}$ (thus fixing a canonical allocation order). It
consumes $\obj_\K$, creates a fresh object, and decrements the class-obligation counters
$\missing{\K}{N}$ and $\mathsf{need}(\K, L)$. Here $\ominus$ denotes saturating subtraction,
i.e., $n \ominus 1 = \max(n-1, 0)$.

Object creation performs (i)~\emph{structural initialization} and (ii)~\emph{scheduling of
later reference choices}: $\initAttrs(\K, \obj_\K)$ populates attributes with well-typed
placeholders (e.g., $\emptyset$ for sets and a distinguished value for single references),
while $\mathsf{addTasks}$ enqueues one task $\refIdx{\obj_\K}{\R}{\K'}{0}{\mathit{REM}}$
per role $\R : \K \to \K'$, where $0$ is the initial rank and $\mathit{REM}$ is a
remaining-rank budget that represents the (finite) bounds-induced choice space. The constraint is strengthened with the
user hook $\specInvOnCreate(\K, \obj_\K)$ and canonicalized by $\simpB$.

\[
\textsc{Obj-Gen}\quad
\state{
    \begin{array}{@{}l@{}}
    \graph{G} \ms \\
    \freshIds{\obj_\K :: \ol} \ms \\
    \mathsf{takenIds}(\mathit{TI}) \ms \\
    \missing{\K}{N} \ms \\
    \mathsf{need}(\K, L) \ms \\
    \refPool{\rtList}
    \ldots
    \end{array}
}{\phiC}
\;\stepRel\;
\state{
    \begin{array}{@{}l@{}}
    \graph{G \ms \mathit{new}} \ms \\
    \freshIds{\ol} \ms \\
    \mathsf{takenIds}(\obj_\K, \mathit{TI}) \ms \\
    \missing{\K}{N \monus 1} \ms \\
    \mathsf{need}(\K, L \monus 1) \ms \\
    \refPool{\rtList'}
    \ldots
    \end{array}
}{\phi'}
\]
\noindent where
\[
\begin{aligned}
\phi'      &= \simpB(\phiC \land \mathsf{hook}_{\mathsf{create}}(\K, \obj_\K)),\\
\rtList'   &= \mathsf{addTasks}(\obj_\K, \K, \mathsf{roleList}(\K), \allIds{\os}, \rtList),\\
\mathit{new} &= \langle \obj_\K : \K \mid \initAttrs(\K, \obj_\K, \allIds{\os}) \rangle.
\end{aligned}
\]


Rule \textsc{Obj-Skip} allows exploring graphs smaller than the maximum scope. It applies only when lower-bound obligations are met ($\missing{\K}{0}$
and $\mathsf{need}(\K, 0)$):

\[
\textsc{Obj-Skip}\quad
\state{
    \begin{array}{@{}l@{}}
    \freshIds{\obj_\K :: \ol} \ms \\
    \missing{\K}{0} \ms \\
    \mathsf{need}(\K, 0)
    \ldots
    \end{array}
}{\phiC}
\;\stepRel\;
\state{
    \begin{array}{@{}l@{}}
    \freshIds{\ol} \ms \\
    \missing{\K}{0} \ms \\
    \mathsf{need}(\K, 0)
    \ldots
    \end{array}
}{\phiC}
\]
\noindent where the constraint is unchanged and the allocation cursor advances by dropping $\obj_\K$.

\begin{remark}[Partial states and global invariants]
\label{rem:partial-states-phase23}\sloppy{
Intermediate states during $\objectBuild$ are \emph{partial}: they do not necessarily
satisfy non-monotonic global invariants $\Inv^{\mathsf{full}}_{\mathsf{graph}}$. Only
monotonic constraints (captured by $\mathsf{specErrorPredGraphPartial}$) are checked
during construction. Full invariants are verified only at acceptance.}
\end{remark}

The $\objectBuild$ phase consists of the repeated application of these symbolic transition
rules until $\freshIds{\nil}$ is exhausted. This exhaustion triggers the transition to
the $\refBuild$ phase. Crucially, as the calculus switches phases, it restricts the global
$\allIds{\os}$ component to exactly the set of identifiers present in $\mathsf{takenIds}$.
This filtering step ensures that the subsequent reference generation phase considers
only the objects that were actually allocated, preventing the creation of dangling
edges to skipped or non-existent objects by construction.


\subsection{Graph Creation Phase ($\refBuild$)}
\label{sec:mmf-reference-generation}

In the $\refBuild$ phase, reference tasks are discharged. A task
$\rt = \refIdx{\obj_\K}{\R}{\K'}{\mathsf{curr}}{\mathit{REM}}$ stores a \emph{rank cursor}
$\mathsf{curr}$ together with a \emph{remaining-rank budget} $\mathit{REM}$.
It represents the interval of admissible ranks
$[\mathsf{curr},\, \mathsf{curr}+\mathit{REM}]$ in the canonical enumeration of role values.

A key semantic choice in $\refBuild$ is that role values are not generated by
incremental extension, but are instead selected as elements of a finite,
\emph{intrinsically defined} domain determined solely by the declared multiplicity
bounds and the identifier universe. Fix a reference task $\refIdx{\obj_\K}{\R}{\K'}{\mathsf{curr}}{\mathit{REM}}$
and let $\os$ be the class-indexed identifier set associated with $\K'$ in $\allIds{\os}$.
Write $P = [o_1, \dots, o_n]$ for the canonical list obtained from $\os$ under a fixed
total order on identifiers, and let $\lb, \ub$ be the lower and upper multiplicity bounds
for role $\R$ (with $0 \le \lb \le \ub \le n$).
The set $\mathcal{D}_{\lb,\ub}(\os)$ of all admissible reference assignments is the union of all subsets of targets
whose size falls within the declared bounds:
\[
\mathcal{D}_{\lb,\ub}(\os) \;=\; \bigcup_{m=\lb}^{\ub} \{\, S \subseteq \os \mid |S| = m \,\},
\qquad
U_{\max} \;=\; |\mathcal{D}_{\lb,\ub}(\os)| - 1.
\]
Initially, a task for this role is created with $\mathsf{curr}=0$ and $\mathit{REM}=U_{\max}$, so that the represented rank interval is exactly $[0, U_{\max}]$.
MMF induces a canonical enumeration over this space without materializing it. The order
is defined hierarchically: first by subset size $m$ (from $\lb$ to $\ub$), and then,
within each size, by lexicographic unranking induced by the list order $P$.

The core mechanism is the unranking function $\mathsf{unrank}(P, m, j)$, which acts as
a stateless cursor. It reconstructs the $j$-th subset of size $m$ directly from $P$
using the standard combinatorial number system for ranking and unranking combinations
(see, e.g., \cite{knuthArtComputerProgramming2011}). The intrinsic selector $\choiceRef$ is
then defined by delegating to the ranked enumeration:
\[
\choiceRef(\os, \K, \R, \K', \mathsf{curr}) \;=\;
\mathsf{rankedFromTo}(P, \lb, \ub, \mathsf{curr}).
\]

Rule \textsc{Ref-Choose} triggers a symbolic transition that commits to the assignment
defined by the current rank $\mathsf{curr}$. It calculates the reference $v$ using
$\choiceRef$ and performs the semantic update to the object's
attribute set:

\begin{align*}
&\textsc{Ref-Choose}\\
&\state{
    \begin{array}{@{}l@{}}
      \allIds{\os} \\
    \graph{\langle \obj_\K : \K \mid \text{AS} \rangle \dots} \ms \\
    \refPool{\rt :: \rtList}
    \ldots
    \end{array}
}{\phiC}
\;\stepRel\;
\state{
    \begin{array}{@{}l@{}}
      \allIds{\os} \\
    \graph{\langle \obj_\K : \K \mid \text{AS}' \rangle \dots} \ms \\
    \refPool{\rtList}
    \ldots
    \end{array}
}{\phi'}
\end{align*}
\noindent where $\rt = \refIdx{\obj_\K}{\R}{\K'}{\mathsf{curr}}{\mathit{REM}}$,
$v = \choiceRef(\os, \K, \R, \K', \mathsf{curr})$,
$\text{AS}' = \setRef(\text{AS}, \R, v)$, and
$\phi' = \simpB(\phiC \land \mathsf{hook}_{\mathsf{setref}}(\K, \R, \obj_\K, v))$.

Rule \textsc{Ref-Skip} represents the decision to \emph{discard} the current assignment
and continue searching. It increments the cursor $\mathsf{curr}$ to explore the
next candidate and decrements the remaining-rank budget $\mathit{REM}$, thereby shrinking
the residual rank interval while preserving the global upper rank $\mathsf{curr}+\mathit{REM}$.

\[
\textsc{Ref-Skip}\quad
\state{
     \begin{array}{@{}l@{}}
     \refPool{\rt :: \rtList}
     \end{array}
     \ldots
}{\phiC}
\;\stepRel\;
\state{
     \begin{array}{@{}l@{}}
     \refPool{\rt' :: \rtList}
     \end{array}
     \ldots
}{\phiC}
\]
\noindent where $\rt = \mathsf{refIdx}(\obj_\K, \R, \K', \mathsf{curr}, \mathit{REM})$, 
$\mathit{REM} > 0$, and $\rt' = \mathsf{refIdx}(\obj_\K, \R, \K', \mathsf{curr}+1, \mathit{REM}-1)$.
Repeated application of \textsc{Ref-Skip} advances through the bounds-induced domain of
admissible values. When $\mathit{REM} = 0$, the only option is to choose (or fail if the
domain is empty), ensuring termination. Completion of $\refBuild$ yields a normal form
$\NF$, checked for global acceptance.

\begin{remark}[Finite enumeration]
\label{rem:finite-enum}
Each reference task has a finite choice space determined by $b$, so termination follows:
\textsc{Ref-Skip} decreases the remaining interval width, and \textsc{Ref-Choose} removes
the task from the pool.
\end{remark}


\section{Symbolic Reachability Engine}
\label{sec:mmf-reachability}


Section~\ref{sec:mmf-semantics} defined the symbolic transition relation
$\stepRel$ over runtime states $\state{\conf}{\phiC}$. We now define a meta-level
exploration strategy that traverses $\stepRel$ under a fixed reduction pipeline comprising exact
deduplication, admissibility pruning, and optional folding.
The exploration engine maintains a loop state
$$L = (\mathit{PQ},\; \mathit{SeenMap},\; \mathit{SeenSuccMap},\; \mathit{ExactMap},\; \mathit{Profiler}),$$
organising the exploration components. The priority queue $\mathit{PQ}$ serves as a worklist of pending nodes, where each node packages a runtime state with an index key and a cached $\NF$ flag. Heuristic search and folding are facilitated by $\mathit{SeenMap}$, which stores retained representatives per index key for coverage-based pruning, and $\mathit{SeenSuccMap}$, which memoises exact successor states to suppress duplicates. Finally, $\mathit{ExactMap}$ collects accepted normal-form states as they are found, while the $\mathit{Profiler}$ tracks instrumentation counters.


\subsection{Meta-level Exploration Semantics}
\label{sec:mmf-exploration}

The search procedure is parameterized by a configuration term \texttt{Opt}, which controls active reduction strategies (e.g., symmetry handling, folding). Starting from an initial state $S_0$, the loop initializes the state components $L$ and iterates as long as the worklist is non-empty. In each step, the engine selects a pending node via \textsc{Select} and delegates it to \textsc{processNode}. The loop terminates when the frontier is exhausted or, in \texttt{findFirst} mode, as soon as a witness is recorded.

The \textsc{processNode} routine bifurcates based on the state's status. If $S$ is a normal form ($\NF(S)$), it constitutes a candidate result; the engine validates it via the mode-specific acceptance predicate and, if successful, records it in $\mathit{ExactMap}$. If $S$ is expandable, the engine computes its one-step symbolic successors using the expansion operator $\textsc{Step}(S)$ (implemented via Maude's reflective \texttt{metaApply}). These raw successors then traverse the reduction pipeline: $\textsc{FilterExact}$ eliminates duplicates against $\mathit{SeenSuccMap}$; $\textsc{Prune}$ discards inadmissible states (violating $\Inv^{\mathsf{part}}_{\mathsf{graph}}$ or with an unsatisfiable constraint); and, if enabled by \texttt{Opt}, a folding stage enforces semantic coverage against $\mathit{SeenMap}$. In the current implementation, folding is applied only to successors that are already in normal form. Finally, surviving successors are added to the worklist via \textsc{Insert}.


\subsection{Subsumption and Folding}
\label{sec:mmf-subsumption}

This subsection formalizes the redundancy elimination mechanism used by the exploration
loop. The key idea is to maintain, among explored normal forms, a \emph{covering antichain}
with respect to a \emph{semantic embedding preorder}. This preorder soundly approximates
the semantic \emph{coverage} relation and supports both queue-time pruning
(\textsc{FoldSift}) and repository maintenance (\textsc{FoldRefresh}).

Recall from Definition~\ref{def:state-semantics} that a runtime state $S \equiv \state{\conf}{\phiC}$ denotes a set $\llbracket S \rrbracket$ of constrained model instances.
Each constrained model instance is a pair $M=\langle \pi(\conf), \rho \rangle$ consisting of a graph topology and a satisfying valuation.
We recall the \emph{semantic coverage preorder} $\preceq$:
\[
S_1 \preceq S_2 \quad\text{iff}\quad \llbracket S_2 \rrbracket \subseteq \llbracket S_1 \rrbracket.
\]
Intuitively, $S_1$ covers $S_2$ if every constrained model instance represented by $S_2$ is also
represented by $S_1$. Coverage is a preorder, not an equivalence.
The engine realizes the preorder $\preceq_{\mathrm{MMF}}$ via $\mathtt{subsumes}(S_1, S_2)$, using three criteria to determine coverage.
First, an \textit{index equality} gate checks $\shapeInd(S_1) = \shapeInd(S_2)$ to filter incompatible candidates.
Second, \textit{structural embedding} employs \texttt{metaMatch} to find a substitution $\sigma$ such that $\pi(\conf_2) =_{E \cup B} \pi(\conf_1) \sigma$.
Third, \textit{constraint entailment} verifies via SMT that $\phiC_2$ entails $\phiC_1 \sigma$.
Meeting these guarantees $S_1$ represents every constrained model instance of $S_2$, making $\mathtt{subsumes}$ a sound (albeit incomplete) semantic coverage test.

Folding maintains a covering antichain $\mathcal{A}$ of normal-form representatives during exploration. Thus, $\mathcal{A}$ contains no mutually subsuming representatives yet covers all explored normal forms. Operationally, $\textsc{FoldSift}$ discards a candidate $S$ only if $\NF(S)$ and some $R \in \mathcal{A}$ satisfies $\mathtt{subsumes}(R, S)$. Conversely, $\textsc{FoldRefresh}$ inserts new representatives by removing any $R \in \mathcal{A}$ subsumed by the incoming state. 
Hence, coverage is sound: every constrained model instance in the bounded semantics remains represented by at least one retained representative in the exploration graph.
MMF does not compute graph isomorphism modulo arbitrary identifier permutations.
Instead, symmetry reduction comes from representing configurations modulo the equational theory $(\Sigma, E \cup B)$ (including ACU for multiset composition) and from shape indexing.
Within a bucket, structural embedding is decided by matching modulo $E \cup B$ via \texttt{metaMatch}, while data-level constraints in $\phiC$ are discharged by SMT.

\section{Correctness of the MMF Calculus}
\label{sec:mmf-proofs}

We summarise the main metatheoretic properties of the bounded construction calculus and its search engine.
The bounded model specification $\mathcal{M}_b$ fixes a finite universe and role ranges, so all guarantees are relative to the declared bounds.
 Soundness ensures that any accepted normal form denotes only constrained model instances that conform to $\mathcal{M}_b$, so reported witnesses are never spurious.
 Completeness ensures that every constrained model instance admitted by $\mathcal{M}_b$ is represented by some retained accepted state, so folding and subsumption do not discard valid solutions.
 Together, these results show that the calculus explores the entire bounded search space modulo coverage-preserving reductions.

Termination and well-formedness explain why exploration is finite and why intermediate states remain meaningful.
Runtime well-formedness $\mathsf{WF}(S)$ captures intrinsic structural consistency of partial graphs and the correct accumulation of constraint obligations in $\phiC$, and each rule preserves it.
Under standard assumptions on equational normalization and finite enumeration of reference tasks, the lexicographic measure on remaining identifiers and task widths strictly decreases.
Consequently, the bounded search terminates and pruning by monotonic structural constraints and SMT unsatisfiability remains correct.
Full statements and proofs appear in Appendix~\ref{sec:app-proofs}.

\section{Evaluation}
\label{sec:evaluation}

We evaluate MMF on the CEO benchmark (Company--Employee--Project) from Section~\ref{sec:mmf-running-example}, which combines structural constraints (e.g., ``CEO has no manager'', acyclicity) with symbolic integer invariants.
The goal is a proof-of-concept assessment of feasibility rather than an exhaustive performance study.
Tables~\ref{tab:eval-modes} and~\ref{tab:eval-ablations} use the bounds from Listing~\ref{lst:org-proj-bounds}: 1 company, 2 employees, and 0--2 projects. Table~\ref{tab:eval-scalability} varies the employee bound from 2 to 5. Timeouts are set at 120\,s.
All runs use a deterministic worklist policy (\texttt{useHeuristicPQ=false}).
We report wall-clock time (\texttt{ms}), Maude rewrites, and driver counters: expanded states (\texttt{popped}), normal forms (\texttt{NF}), successors generated/enqueued, and eliminations by exact deduplication, structural pruning (\texttt{prunedBad}), SMT unsatisfiability (\texttt{prunedUnsat}), and folding (\texttt{folded}).
Unless stated otherwise, the baseline enables SMT pruning, folding, and shape indexing (\texttt{useSymbContraintFiltering=true}, \texttt{useFolding=true}, \texttt{useIndex=true}), and removes redundant solutions at insertion (\texttt{removeRedundantSolutions=true}).

Table~\ref{tab:eval-modes} compares \texttt{findAll}/\texttt{findFirst} and \texttt{find}/\texttt{check} under the baseline pipeline.
For \texttt{findAll}, more than half of generated successors are eliminated before queueing, with structural pruning being the largest contributor and SMT pruning and folding providing additional reductions.
\texttt{findAll(check)} expands fewer states than \texttt{findAll(find)} but is slightly slower, indicating higher per-state overhead in \texttt{check} (property-side evaluation and constraint processing).
For \texttt{findFirst}, folding does not contribute on this scope (\texttt{folded}=0); runtime is largely determined by pruning and successor generation.

\begin{table}[t]
\small
\centering
\setlength{\tabcolsep}{4pt}
\begin{tabular}{lrrrrrr}
\hline
run & ms & rewrites & popped & NF & pruned & folded \\
\hline
\texttt{findAll(find)}   & 2181 & 7.92M  & 450 & 52 & 354 & 88 \\
\texttt{findFirst(find)} &  136 & 284.8K &  29 &  1 &  27 &  0 \\
\texttt{findAll(check)}  & 2409 & 11.11M & 374 & 40 & 298 & 64 \\
\texttt{findFirst(check)}&  144 & 317.6K &  29 &  1 &  21 &  0 \\
\hline
\end{tabular}
\caption{Mode comparison under the baseline pipeline. \emph{pruned} = \texttt{prunedBad}+\texttt{prunedUnsat}.}
\label{tab:eval-modes}
\end{table}

Table~\ref{tab:eval-ablations} isolates the reduction mechanisms for \texttt{findAll(find)}.
Disabling SMT pruning while keeping folding enabled times out, whereas disabling both SMT pruning and folding completes but is an order of magnitude slower than baseline; thus SMT satisfiability filtering is essential and also stabilizes folding (which relies on entailment checks).
Disabling folding decreases runtime but increases the number of normal forms; redundant solutions are then filtered only at insertion time, showing that folding prevents exploring many redundant candidates to completion.
Folding without indexing times out, indicating that shape indexing is a prerequisite for making folding feasible.

\begin{table}[t]
\small
\centering
\setlength{\tabcolsep}{4pt}
\begin{tabular}{lrrrrrr}
\hline
run & ms & rewrites & popped & NF & inserted & redundant \\
\hline
\texttt{A1-baseline}         &  2185 & 7.92M  &  450 &  52 &  52 &  0 \\
\texttt{A2-noSMT}            & timeout &  -   &   -  &  -  &  -  &  - \\
\texttt{A3-noFolding}        &  1258 & 2.20M  &  494 &  96 &  52 & 44 \\
\texttt{A4-folding-noIndex}  & timeout &  -   &   -  &  -  &  -  &  - \\
\texttt{A5-noSMT-noFolding}  & 20299 & 28.01M & 2054 & 192 & 100 & 92 \\
\hline
\end{tabular}
\caption{Ablations for \texttt{findAll(find)}. Timeouts occur at 120\,s.}
\label{tab:eval-ablations}
\end{table}

Table~\ref{tab:eval-scalability} scales \texttt{findFirst(check)} from 2 to 5 employees, comparing baseline to \texttt{noFolding}.
Baseline time increases sharply with the bound, while \texttt{noFolding} scales much more gently; at 5 employees, \texttt{noFolding} is over an order of magnitude faster.
Since folding eliminates only a few successors at these bounds, this suggests that subsumption checks dominate cost in witness-oriented search when folding successes are rare.
All \texttt{folding-noIndex} scalability runs timed out, consistent with the ablation results: indexing is necessary for folding to be viable.

Overall, structural pruning removes most successors in exhaustive search, while SMT pruning is essential for feasibility and stabilises folding.
Folding reduces redundant candidates when shape indexing makes subsumption checks practical.
In \texttt{check} mode, witness search incurs higher per-state overhead but remains practical on the explored bounds.

\begin{table}[t]
\small
\centering
\setlength{\tabcolsep}{4pt}
\begin{tabular}{rrrrrr}
\hline
employees & ms (base) & rewrites (base) & folded (base) & ms (noFold) & rewrites (noFold) \\
\hline
2 &   97 & 0.10M & 0 &  88 & 0.05M \\
3 &  147 & 0.32M & 0 & 107 & 0.05M \\
4 &  737 & 4.08M & 3 & 172 & 0.13M \\
5 & 3098 & 20.27M& 6 & 264 & 0.24M \\
\hline
\end{tabular}
\caption{Scalability for \texttt{findFirst(check)} as the employee bound increases.}
\label{tab:eval-scalability}
\end{table}

The evaluation focuses on one benchmark family; folding effectiveness may differ for specifications with stronger symmetries or weaker constraints.
We report single-run timings; although runs are deterministic, wall time can vary with system load.
Rewrite counts are a stable internal proxy for work but do not isolate SMT solver time.
Timeout-based conclusions depend on the 120\,s cutoff, but the observed blow-ups are consistent across the ``noSMT+folding'' and ``folding-noIndex'' configurations.

\section{Related work}
\label{sec:related-work}

\paragraph{Bounded Model Finding.}
MMF shares the goal of bounded model finders like Alloy \cite{jackson_2019_alloy_language_tool} and USE \cite{kuhlmann_2011_extensive_validation_ocl,gogolla_0000_model_finding_model}: generating concrete object graphs that satisfy structural constraints. Unlike these tools, which typically rely on translation to SAT/SMT solvers (often creating a semantic gap), MMF operates as a native symbolic rewriting engine within rewriting logic. This allows domain logic and constraints to be defined directly in Maude \cite{durn_2020_programming_symbolic_computation}, providing a uniform environment for specification, analysis, and verification.

\paragraph{Symbolic Reachability and Symbolic Rewriting.}
Symbolic analysis in Maude ranges from narrowing-based reachability to symbolic rewriting modulo SMT~\cite{meseguer_2007_symbolic_reachability_analysis,bae_2019_symbolic_state_space,lpez-rueda_2023_efficient_canonical_narrowing}.
MMF repurposes this symbolic infrastructure for bounded structural model finding.
Technically, MMF explores a constrained symbolic rewriting system over runtime states that refine partial object graphs, rather than using narrowing to solve for substitutions that drive temporal reachability.
The correctness of our calculus rests on established results in this domain.
Our \emph{Soundness} theorem (Theorem~\ref{thm:soundness}) mirrors the simulation results of Bae and Rocha \cite{bae_2019_symbolic_state_space} and Arias et al. \cite{arias_2024_rewritinglogicwithsmtbased_formal_analysis}, ensuring that every symbolic step corresponds to a valid refinement of the concrete model space.
Our \emph{Completeness} theorem (Theorem~\ref{thm:completeness}) follows from a bounded two-phase enumeration argument within the user-declared bounds, guaranteeing that our generation strategy can construct any satisfying concrete model instance that exists within those bounds.

\paragraph{State Space Folding.}
To ensure feasibility, MMF employs state folding to prune redundant search paths. This technique relies on a subsumption relation that approximates semantic containment. As formalized by Arias et al. \cite{arias_2024_rewritinglogicwithsmtbased_formal_analysis} and Meseguer \cite{meseguer_2026_symbolic_computation_verification}, a correctly defined subsumption relation ($\preceq$) ensures that if a state $S$ is subsumed by an existing representative $S'$, pruning $S$ preserves completeness because all constrained model instances denoted by $S$ are already covered by $S'$. MMF implements this via graph shape indexing and logical entailment checks.

\paragraph{Comparison with specific approaches.}
Riesco's testing techniques exploit Maude's symbolic engines (notably narrowing) to generate concrete witnesses~\cite{riesco_using_2012}.
However, Riesco targets test input generation for functional correctness of operational semantics, whereas MMF focuses on structural model finding under schema constraints, using SMT integration primarily for satisfiability of data invariants rather than full symbolic execution.
 Escobar et al. develop folding optimisations for SMT-based symbolic reachability based on narrowing~\cite{escobar_symbolic_2023}.
 MMF adapts these principles but specializes the redundancy layer (shape indexing, ACU-matching, and entailment-based subsumption) to exploit order-only symmetries induced by the equational theory (notably ACU) and to reduce redundancy in strictly bounded structural generation.

\section{Conclusion and Future Work}
\label{sec:conclusion}

This paper presented the Maude Model Finder (MMF), a native approach to bounded model finding in the Maude ecosystem. By formalizing bounded model generation as symbolic transition exploration over runtime states within rewriting logic, MMF allows users to validate structural and data-level constraints without departing from Maude's framework. Our approach synthesizes the structural generation techniques of SAT-based finders with symbolic rewriting modulo SMT, using SMT strictly as a background solver for constraint satisfiability and entailment, and leveraging folding to quotient redundant symbolic states. We establish termination, soundness, and completeness of the bounded construction within the declared bounds, and justify folding via a coverage-preserving subsumption test.

Several avenues for future work remain. First, we plan to generalize the current meta-level exploration engine by integrating Maude's strategy language \cite{eker_2023_maude_strategy_language}, offering users a declarative way to customize search heuristics and pruning policies. Second, while our current symmetry reduction relies on canonical construction and post-hoc folding, we aim to investigate constructive symmetry breaking techniques to further prune the frontier. Finally, we intend to expand the supported constraint language to cover a richer subset of OCL, enhancing MMF's utility for MDE applications.

\section*{Acknowledgements}
The author is grateful to Jos\'{e} Meseguer and Santiago Escobar for insightful discussions and detailed feedback on drafts of this article.

\bibliographystyle{splncs04}
\bibliography{bib.bib}

@article{arias_2024_rewritinglogicwithsmtbased_formal_analysis,
  author   = {Arias, Jaime and Bae, Kyungmin and Olarte, Carlos and Ölveczky, Peter Csaba and Petrucci, Laure},
  doi      = {10.3233/FI-242195},
  issn     = {0169-2968},
  journal  = {Fundamenta Informaticae},
  language = {EN},
  month    = {September},
  note     = {Publisher: SAGE Publications},
  number   = {3-4},
  pages    = {261--312},
  title    = {A {Rewriting}-logic-with-{SMT}-based {Formal} {Analysis} and {Parameter} {Synthesis} {Framework} for {Parametric} {Time} {Petri} {Nets}},
  volume   = {192},
  year     = {2024}
}

@phdthesis{bae_2014_rewritingbased_model_checking,
  author    = {Bae, Kyungmin},
  copyright = {Copyright 2014 Kyungmin Bae},
  month     = {September},
  school    = {University of Illinois at Urbana-Champaign},
  title     = {Rewriting-based model checking methods},
  url       = {https://hdl.handle.net/2142/50553},
  urldate   = {2025-10-27},
  year      = {2014}
}

@article{bae_2019_symbolic_state_space,
  author   = {Bae, Kyungmin and Rocha, Camilo},
  doi      = {10.1016/j.scico.2019.03.006},
  issn     = {0167-6423},
  journal  = {Science of Computer Programming},
  month    = {June},
  pages    = {20--42},
  title    = {Symbolic state space reduction with guarded terms for rewriting modulo {SMT}},
  volume   = {178},
  year     = {2019}
}

@article{durn_2020_programming_symbolic_computation,
  author   = {Durán, Francisco and Eker, Steven and Escobar, Santiago and Martí-Oliet, Narciso and Meseguer, José and Rubio, Rubén and Talcott, Carolyn},
  doi      = {10.1016/j.jlamp.2019.100497},
  issn     = {2352-2208},
  journal  = {Journal of Logical and Algebraic Methods in Programming},
  month    = {January},
  pages    = {100497},
  title    = {Programming and symbolic computation in {Maude}},
  volume   = {110},
  year     = {2020}
}

@article{eker_2023_maude_strategy_language,
  author   = {Eker, Steven and Martí-Oliet, Narciso and Meseguer, José and Rubio, Rubén and Verdejo, Alberto},
  doi      = {10.1016/j.jlamp.2023.100887},
  issn     = {2352-2208},
  journal  = {Journal of Logical and Algebraic Methods in Programming},
  month    = {August},
  pages    = {100887},
  title    = {The {Maude} strategy language},
  volume   = {134},
  year     = {2023}
}

@article{meseguer_2007_symbolic_reachability_analysis,
  author   = {Meseguer, José and Thati, Prasanna},
  doi      = {10.1007/s10990-007-9000-6},
  issn     = {1573-0557},
  journal  = {Higher-Order and Symbolic Computation},
  language = {en},
  month    = {June},
  number   = {1},
  pages    = {123--160},
  title    = {Symbolic reachability analysis using narrowing and its application to verification of cryptographic protocols},
  volume   = {20},
  year     = {2007}
}

@article{gogolla_0000_model_finding_model,
  author    = {Martin Gogolla and
               Loli Burgue{\~{n}}o and
               Antonio Vallecillo},
  editor    = {Regina Hebig and
               Thorsten Berger},
  title     = {Model Finding and Model Completion with {USE}},
  booktitle = {Proceedings of {MODELS} 2018 Workshops: ModComp, MRT, OCL, FlexMDE,
               EXE, COMMitMDE, MDETools, GEMOC, MORSE, MDE4IoT, MDEbug, MoDeVVa,
               ME, MULTI, HuFaMo, AMMoRe, {PAINS} co-located with {ACM/IEEE} 21st
               International Conference on Model Driven Engineering Languages and
               Systems {(MODELS} 2018), Copenhagen, Denmark, October, 14, 2018},
  series    = {{CEUR} Workshop Proceedings},
  volume    = {2245},
  pages     = {194--200},
  publisher = {CEUR-WS.org},
  year      = {2018},
  url       = {https://ceur-ws.org/Vol-2245/ocl\_paper\_9.pdf},
  bibsource = {dblp computer science bibliography, https://dblp.org}
}

@article{jackson_2019_alloy_language_tool,
  author     = {Jackson, Daniel},
  doi        = {10.1145/3338843},
  issn       = {0001-0782},
  journal    = {Commun. ACM},
  month      = {August},
  number     = {9},
  pages      = {66--76},
  shorttitle = {Alloy},
  title      = {Alloy: a language and tool for exploring software designs},
  volume     = {62},
  year       = {2019}
}

@inproceedings{kuhlmann_2011_extensive_validation_ocl,
  address   = {Berlin, Heidelberg},
  author    = {Kuhlmann, Mirco and Hamann, Lars and Gogolla, Martin},
  booktitle = {Objects, {Models}, {Components}, {Patterns}},
  doi       = {10.1007/978-3-642-21952-8_21},
  editor    = {Bishop, Judith and Vallecillo, Antonio},
  isbn      = {978-3-642-21952-8},
  language  = {en},
  pages     = {290--306},
  publisher = {Springer},
  title     = {Extensive {Validation} of {OCL} {Models} by {Integrating} {SAT} {Solving} into {USE}},
  year      = {2011}
}

@inproceedings{meseguer_2026_symbolic_computation_verification,
  author    = {Meseguer, José},
  title     = {Symbolic Computation and Verification Methods in Maude},
  booktitle = {Logic-Based Program Synthesis and Transformation - 34th International Symposium, {LOPSTR} 2025, Rende, Italy, September 9-10, 2025, Proceedings},
  series    = {Lecture Notes in Computer Science},
  year      = {2025},
  publisher = {Springer}
}

@book{clavel_2007_all_about_maude,
  author    = {Manuel Clavel and Francisco Dur{\'a}n and Steven Eker and Patrick Lincoln and Narciso Mart{\'i}-Oliet and Jos{\'e} Meseguer and Carolyn Talcott},
  title     = {All About Maude - A High-Performance Logical Framework: How to Specify, Program, and Verify Systems in Rewriting Logic},
  year      = {2007},
  publisher = {Springer},
  address   = {Berlin, Heidelberg},
  isbn      = {978-3-540-40212-1},
  series    = {Lecture Notes in Computer Science},
  volume    = {4350},
  doi       = {10.1007/978-3-540-40212-1}
}

@inproceedings{BoronatSLE17,
  author    = {Artur Boronat},
  title     = {Structural model subtyping with {OCL} constraints},
  booktitle = {Proceedings of the 10th {ACM} {SIGPLAN} International Conference on
               Software Language Engineering, {SLE} 2017, Vancouver, BC, Canada,
               October 23-24, 2017},
  pages     = {194--205},
  year      = {2017},
  publisher = {{ACM}}
}

@article{boronat_safe_2023,
  title    = {Safe reuse in modelling language engineering using model subtyping with {OCL} constraints},
  volume   = {22},
  issn     = {1619-1374},
  doi      = {10.1007/s10270-022-01028-7},
  language = {en},
  number   = {3},
  journal  = {Software and Systems Modeling},
  author   = {Boronat, Artur},
  month    = jun,
  year     = {2023},
  pages    = {797--818}
}

@inproceedings{riesco_using_2012,
  address   = {Berlin, Heidelberg},
  title     = {Using {Narrowing} to {Test} {Maude} {Specifications}},
  isbn      = {978-3-642-34005-5},
  doi       = {10.1007/978-3-642-34005-5_11},
  language  = {en},
  booktitle = {Rewriting {Logic} and {Its} {Applications}},
  publisher = {Springer},
  author    = {Riesco, Adrián},
  editor    = {Durán, Franciso},
  year      = {2012},
  pages     = {201--220}
}

@article{lpez-rueda_2023_efficient_canonical_narrowing,
  title    = {An efficient canonical narrowing implementation with irreducibility and {SMT} constraints for generic symbolic protocol analysis},
  volume   = {135},
  issn     = {2352-2208},
  doi      = {10.1016/j.jlamp.2023.100895},
  journal  = {Journal of Logical and Algebraic Methods in Programming},
  author   = {López-Rueda, Raúl and Escobar, Santiago and Sapiña, Julia},
  month    = oct,
  year     = {2023},
  pages    = {100895}
}

@inproceedings{escobar_symbolic_2023,
  address   = {New York, NY, USA},
  series    = {{FTSCS} 2023},
  title     = {Symbolic {Analysis} by {Using} {Folding} {Narrowing} with {Irreducibility} and {SMT} {Constraints}},
  isbn      = {979-8-4007-0398-0},
  doi       = {10.1145/3623503.3623537},
  booktitle = {Proceedings of the 9th {ACM} {SIGPLAN} {International} {Workshop} on {Formal} {Techniques} for {Safety}-{Critical} {Systems}},
  publisher = {Association for Computing Machinery},
  author    = {Escobar, Santiago and López-Rueda, Raúl and Sapiña, Julia},
  month     = oct,
  year      = {2023},
  pages     = {14--25}
}

@book{knuthArtComputerProgramming2011,
  title = {The Art of Computer Programming. {{Volume}} 4a, {{Part}} 1: {{Combinatorial}} Algorithms},
  shorttitle = {The Art of Computer Programming. {{Volume}} 4a, {{Part}} 1},
  author = {Knuth, Donald Ervin},
  year = 2011,
  publisher = {Addison-Wesley},
  address = {Boston, Mass. ; London},
  isbn = {978-0-201-03804-0},
  lccn = {005.1}
}

\appendix
\renewcommand*{\theHsection}{\thesection}
\renewcommand*{\theHsubsection}{\thesection.\arabic{subsection}}
\section{Correctness of the MMF Calculus}
\label{sec:app-proofs}

This appendix gives the full statements and proofs for the results summarised in Section~\ref{sec:mmf-proofs}.

\subsection{Soundness and Completeness}
\label{sec:mmf-sound-complete}

\begin{theorem}[Soundness of MMF calculus]
\label{thm:soundness}
If $S_0 = \mathsf{init}(\mathcal{M}_b, \texttt{find})$ and $S_0 \stepRel^{*} S$ with
$\Accept_{\mathtt{find}}(S)$, then $\llbracket S \rrbracket \subseteq \llbracket \mathcal{M}_b \rrbracket$.
\end{theorem}

\begin{proof}
The proof relies on the fact that MMF construction preserves a runtime validity invariant.
The \emph{object generation phase} (Item 1) preserves boundedness constraints because \textsc{Obj-Gen} consumes one preallocated identifier and adds a valid object skeleton, while \textsc{Obj-Skip} is guarded by lower-bound obligations.
All attribute constraints triggered by creation are accumulated in $\phiC$ via the $\mathsf{hook}_{\mathsf{create}}$ mechanism.
The \emph{reference generation phase} (Item 2) respects role domains because \textsc{Ref-Choose} commits values from the bounded universe defined by the finite sort inhabitants, and it concurrently accumulates reference constraints in $\phiC$ via $\mathsf{hook}_{\mathsf{setref}}$.
Finally (Item 3), an accepted normal form $S$ implies $\NF(S)$ and $\SAT(\phiC)$.
If we pick a witness valuation $\rho$ such that $\rho \models \phiC$, the constrained model instance $\langle \pi(\conf), \rho \rangle$ satisfies all structural invariants (by $\Accept_{\mathtt{find}}(S)$) and all data invariants (because $\models (\phiC \Rightarrow \Inv_{\mathsf{attr}}(\pi(\conf)))$).
Thus, $\langle \pi(\conf), \rho \rangle$ is a valid member of the semantics that conforms to $\mathcal{M}_b$.
\end{proof}

We say that an accepted state $S$ is \emph{retained} if it belongs to the covering antichain $\mathcal{A}$ computed by the engine (see Section~\ref{sec:mmf-subsumption}).

\begin{theorem}[Completeness of MMF calculus]
\label{thm:completeness}
For every constrained model instance $\langle G, \rho \rangle \in \llbracket \mathcal{M}_b \rrbracket$, there exists a retained accepted state $S$ reachable from $\mathsf{init}(\mathcal{M}_b, \texttt{find})$ such that $\langle G, \rho \rangle \in \llbracket S \rrbracket$.
\end{theorem}

\begin{proof}
The completeness argument constructs a derivation guided by the target graph $G$.
First, since the bounds $b$ preallocate sufficient identifiers, we can select a sequence of \textsc{Obj-Gen} and \textsc{Obj-Skip} steps along the canonical identifier order that allocates exactly the identifier set of $G$.
Second, for every role assignment in $G$, let $v$ be the specific set of target identifiers assigned to role $\R$.
By the definition of the reference phase (Section~\ref{sec:mmf-semantics}), the domain of valid assignments is the finite set $\mathcal{D}_{\lb,\ub}(\os)$ comprising all subsets of the available universe within the multiplicity bounds.
Since MMF imposes a canonical total order on this domain (via lexicographic unranking) and the search interval covers the entire domain, there exists a unique index $k$ such that $\mathsf{unrank}(P, \dots, k) = v$.
The sequence of \textsc{Ref-Skip} operations (which increment the cursor) followed by \textsc{Ref-Choose} (which selects the current cursor) guarantees that the specific index $k$ is reachable, thus committing to the exact target values of $G$.
Because $\langle G, \rho \rangle$ is a valid constrained model instance, $\rho$ satisfies $\Inv_{\mathsf{attr}}(G)$, and the final constraint $\phiC$ is satisfied by $\rho$ (i.e., $\rho \models \phiC$).
Thus, there exists a reachable accepted state $S_{exact}$ such that $\langle G, \rho \rangle \in \llbracket S_{exact} \rrbracket$.
Finally, by the coverage soundness of the folding strategy (Section~\ref{sec:mmf-subsumption}), the engine maintains a covering antichain $\mathcal{A}$ such that for any such reachable $S_{exact}$, there exists a retained representative $S \in \mathcal{A}$ with $S \preceq S_{exact}$.
By definition of $\preceq$, this implies $\llbracket S_{exact} \rrbracket \subseteq \llbracket S \rrbracket$, and therefore $\langle G, \rho \rangle \in \llbracket S \rrbracket$.
\end{proof}

\subsection{Termination of the Bounded Construction}
\label{sec:mmf-termination}

Fix a bounded model specification $\mathcal{M}_b$.
Assume:

\begin{assumption}[Well-defined normalization on reachable ground states]
\label{ass:eq-well-defined}
Let $E$ be terminating and confluent modulo $B$ on the class of ground terms reachable
from the initial state $S_0$ by $\stepRel$.
Moreover, all auxiliary defined functions
used by the construction rules (including graph-update helpers, task constructors, and
$\simpB$) are ground-sufficiently complete on those reachable arguments, so that post-step
canonicalization yields a unique $E/B$-normal form.
\end{assumption}

\begin{assumption}[Finite bounded enumeration]
\label{ass:finite-enum}
For each reference-task token in $\refPool{\rtList}$, the corresponding admissible value
domain induced by $\mathsf{boundsDecl}$ and $\allIds{\os}$ is finite, and the task interval
stored in the token ranges over that domain without repetition.
\end{assumption}

The calculus manipulates \emph{partial} graphs, so we do not aim to preserve global
constraints at every intermediate step.
Instead, we preserve a single runtime notion
of well-formedness that captures (i) intrinsic partial correctness and (ii) correct
accumulation of the constraint.

\begin{definition}[Model well-formedness]
\label{def:wf}
Define $\mathsf{WF}(\state{\conf}{\phiC})$ as
\(
\mathsf{OkConf}(\conf) \;\wedge\; \mathsf{OkConstraint}(\conf, \phiC),
\)
where:
\begin{enumerate}
    \item $\mathsf{OkConf}(\conf)$ states that $\conf$ is a well-typed configuration intrinsically consistent with the schema.
    Operationally, it requires that $\mathsf{specErrorPredGraphPartialBuiltin}(\conf)$ is false (ensuring unique object identifiers and structural bound consistency) and that $\mathsf{specErrorPredGraphPartial}(\conf)$ is false (ensuring monotonic user-defined graph constraints).
    It also implies that the presence of obligation tokens is synchronized with the construction phase.
    \item $\mathsf{OkConstraint}(\conf, \phiC)$ states that $\phiC$ is an $E/B$-canonical Boolean constraint over attribute variables keyed by the (possibly symbolic) object identifiers allocated in $C$.
    Formally, it equals (modulo $\simpB$ and $E/B$) the conjunction of all $\mathsf{hook}_\ell$ constraints implied by the construction events in $\conf$.
\end{enumerate}
\end{definition}

\begin{remark}[Partial states and global invariants]
\label{rem:partial-states}
Intermediate MMF states are intended to be \emph{partial} and are not required to satisfy
non-monotonic global constraints.
Therefore, we cannot claim that every rule application yields a valid model.
Instead, we prove that rules preserve \emph{runtime well-formedness} ($\mathsf{WF}$), which ensures intrinsic structural consistency and the correct accumulation of constraint obligations.
\end{remark}

Next, we define a decreasing measure that ensures termination.

\begin{definition}[Bounded progress measure]
\label{def:phase-measure}
Define a measure $\mu(\conf)$ on runtime configurations by a lexicographic pair:
\[
\mu(\conf) \;=
\bigl(\;|\freshIds{\conf}|,\;\;\sum_{\rt \in \refPool{\conf}} \mathsf{width}(\rt)\;\bigr)
\in \mathbb{N} \times \mathbb{N},
\]
where $|\freshIds{\conf}|$ is the length of the identifier pool and
$\mathsf{width}(\rt)$ is the remaining interval size for a task $\rt$
(e.g., $\mathit{REM} + 1$ for $\rt = \refIdx{o}{r}{K'}{\mathsf{curr}}{\mathit{REM}}$).
\end{definition}

\begin{lemma}[Termination of $\stepRel$ on MMF states]
\label{lem:termination}
Every $\stepRel$-derivation from $S_0$ is finite.
\end{lemma}

\begin{proof}
By construction, each object-phase rule strictly decreases $|\freshIds{\cdot}|$ by consuming one identifier (\textsc{Obj-Gen} or \textsc{Obj-Skip}).
Each reference-phase \textsc{Ref-Skip} decreases $\mathsf{width}(\rt)$ for the head task, and \textsc{Ref-Choose} removes one task from the pool.
By Assumption~\ref{ass:finite-enum}, each task has finite width.
Hence $\mu(\conf)$ strictly decreases in the lexicographic order on $\mathbb{N} \times \mathbb{N}$, so no infinite derivation exists.
\end{proof}

\begin{lemma}[One-step preservation of well-formedness]
\label{lem:preservation}
If $\mathsf{WF}(S)$ and $S \stepRel S'$, then $\mathsf{WF}(S')$.
\end{lemma}

\begin{proof}
Proceed by cases on the applied rule schema, grouped by phase.
\begin{enumerate}
\item \emph{Object generation rules.}
For \textsc{Obj-Gen}, $\mathsf{OkConf}$ is preserved because a fresh identifier is consumed and a well-typed object skeleton is added.
The new constraint is exactly $\simpB(\phiC \wedge \mathsf{hook}_{\mathsf{create}}(\cdots))$ by definition of lifting, so $\mathsf{OkConstraint}$ holds.
For \textsc{Obj-Skip}, only the identifier pool changes and $\mathsf{OkConf}$ is preserved.

\item \emph{Reference generation rules.}
For \textsc{Ref-Choose}, $\mathsf{OkConf}$ is preserved because the committed value is drawn from the allocated universe and respects multiplicity bounds.
The lifted update with $\mathsf{hook}_{\mathsf{setref}}(\cdots)$ preserves $\mathsf{OkConstraint}$.
For \textsc{Ref-Skip}, $\mathsf{OkConf}$ is unchanged and the task interval shrinks consistently.
\end{enumerate}

Assumptions~\ref{ass:eq-well-defined}--\ref{ass:finite-enum} justify that the post-step canonicalization used in the definition of $\stepRel$ terminates and yields the intended constructor representation on reachable ground terms.
\end{proof}

\end{document}